\documentclass[9pt]{article}
\usepackage{amsfonts}
\usepackage{times}
\usepackage{graphicx}
\usepackage{latexsym}
\usepackage{dsfont}
\usepackage{amssymb}
\usepackage{amsmath}
\usepackage{cite}
\usepackage{verbatim}
\usepackage{subfigure}
\newtheorem{theorem}{Theorem}

\newtheorem{algorithm}[theorem]{Algorithm}

\newtheorem{proposition}[theorem]{Proposition}

\newenvironment{proof}{ \textbf{Proof:} }{ \hfill $\Box$}

\newcommand{\figref}[1]{{Fig.}~\ref{#1}}

\def\bb0{{\mathbb{0}}}

\def\bb{{\mathbf{b}}}

\def\bff{{\mathbf{f}}}

\def\bn{{\mathbf{n}}}

\def\bs{{\mathbf{s}}}

\def\by{{\mathbf{y}}}

\def\b0{{\mathbf{0}}}

\def\bA{{\mathbf{A}}}

\def\bE{{\mathbf{E}}}
\def\bF{{\mathbf{F}}}

\def\bH{{\mathbf{H}}}
\def\bI{{\mathbf{I}}}

\def\bP{{\mathbf{P}}}

\def\bT{{\mathbf{T}}}
\def\bU{{\mathbf{U}}}
\def\bV{{\mathbf{V}}}
\def\bW{{\mathbf{W}}}

\def\bbE{{\mathbb{E}}}

\def\cF{\mathcal{F}}

\def\cI{\mathcal{I}}

\def\cN{\mathcal{N}}

\def\cU{\mathcal{U}}

\def\sf0{{\mathsf{0}}}

\usepackage{spconf,amsmath,graphicx}
\newcommand{\sref}[1]{{Section}~\ref{#1}}
\usepackage{epstopdf}
\usepackage{varwidth}
\usepackage{enumerate}
\usepackage{algorithmicx}
\usepackage{algorithm}
\usepackage{float}
\usepackage{color}
\usepackage{makeidx}
\usepackage{bbm}
\usepackage{makecell}

\allowdisplaybreaks
\def\j{\mathrm{j}}
\begin{document}
\title{Gram Schmidt Based Greedy Hybrid Precoding for Frequency Selective Millimeter Wave MIMO Systems}
\name{Ahmed Alkhateeb and Robert W. Heath Jr. \thanks{This work is supported in part by the National Science Foundation under Grant No. 1319556, and by a gift from Nokia.}}
\address{The University of Texas at Austin, TX, USA, Email: $\{$aalkhateeb,  rheath$\}$@utexas.edu}
\maketitle
\ninept

\begin{abstract}
Hybrid analog/digital precoding allows millimeter wave MIMO systems to leverage large antenna array gains while permitting low cost and power consumption hardware. Most prior work has focused on hybrid precoding for narrow-band mmWave systems. MmWave systems, however, will likely operate on wideband channels with frequency selectivity. Therefore, this paper considers frequency selective hybrid precoding with RF beamforming vectors taken from a quantized codebook. For this system, a low-complexity yet near-optimal greedy algorithm is developed for the design of the hybrid analog/digital precoders. The proposed algorithm greedily selects the RF beamforming vectors using Gram-Schmidt orthogonalization. Simulation results show that the developed precoding design algorithm achieves very good performance compared with the unconstrained solutions while requiring less complexity.
\end{abstract}


\section{Introduction} \label{sec:Intro}
Millimeter wave (mmWave) communication can leverage the large bandwidth potentially available at the high frequency bands to provide high data rates \cite{Nitsche2014,Pi2011,Rappaport2014,Boccardi2014,Bai2014}. To guarantee  sufficient received signal power  at these high frequencies, though, large antenna arrays need to be deployed at both the transmitter and receiver \cite{Pi2011,Rappaport2014,Rappaport2013a}. Designing precoding and combining matrices for these mmWave wideband large MIMO systems differs from lower-frequency solutions. This is mainly due to the different hardware constraints on the mixed signal components because of their high cost and power consumption \cite{Alkhateeb2014d}. Therefore, developing precoding schemes for wideband mmWave systems is important for building these systems.

For the sake of low power consumption, hybrid analog/digital precoding solutions, that divide the precoding between analog and digital domains, and hence requiring smaller number of RF chains, were proposed in \cite{Zhang2005,Venkateswaran2010,ElAyach2014,Alkhateeb2013,Alkhateeb2014,Sohrabi2015,Mendez-Rial2015a,Chen2015,Kim2013}.  For general MIMO systems, hybrid precoding design with diversity and spatial multiplexing objectives were investigated in \cite{Zhang2005,Venkateswaran2010}. In \cite{ElAyach2014}, the sparse nature of mmWave channels was exploited, and low-complexity iterative algorithms based on matching pursuit were devised, assuming perfect channel knowledge at the transmitter. Extensions to the case when only partial channel knowledge is required was considered in \cite{Alkhateeb2013,Alkhateeb2014}. Other heuristic algorithms that do not rely on orthogonal matching pursuit were also proposed in \cite{Sohrabi2015,Mendez-Rial2015a,Chen2015} for the hybrid precoding design with perfect channel knowledge at the transmitter. The work in \cite{ElAyach2014,Alkhateeb2013,Alkhateeb2014,Sohrabi2015,Mendez-Rial2015a} assumed a narrow-band mmWave channel, with perfect or partial channel knowledge at the transmitter. In \cite{Kim2013}, hybrid beamforming with only a single-stream transmission over MIMO-OFDM system was considered. The solution in \cite{Kim2013} though relied on the joint exhaustive search over both RF and baseband codebooks which results in high-complexity. As mmWave communication is expected to employ broadband channels, developing spatial multiplexing hybrid precoding algorithms for wideband mmWave systems is important. 

In this paper, we investigate the frequency selective hybrid precoding design to maximize the achievable mutual information given that the RF precoders are taken from a quantized codebook. We first derive the optimal baseband precoders as functions of the RF precoders. Then, we design a greedy hybrid precoding algorithm based on Gram-Schmidt orthogonalization. Despite its low-complexity, the proposed algorithm is illustrated to achieve a similar performance compared with the optimal hybrid precoding design that  requires an exhaustive search over the RF codebooks.
\section{System Model} \label{sec:Model}
Consider the OFDM based system model in \figref{fig:Model} where a basestation (BS) with $N_\mathrm{BS}$ antennas and $N_\mathrm{RF}$ RF chains is assumed to communicate with a single mobile station (MS) with $N_\mathrm{MS}$ antennas and $N_\mathrm{RF}$ RF chains. The BS and MS communicate via $N_\mathrm{S}$ length-$K$ data symbol blocks, such that $N_\mathrm{S} \leq N_\mathrm{RF} \leq N_\mathrm{BS}$ and $N_\mathrm{S}\leq N_\mathrm{RF} \leq N_\mathrm{MS}$. 

At the transmitter, the $N_\mathrm{S}$ data symbols $\bs_k$ at each subcarrier $k=1, ..., K$ are first precoded using an $N_\mathrm{RF} \times N_\mathrm{S}$ digital precoding matrix $\bF[k]$, and the symbol blocks are transformed to the time-domain using $N_\mathrm{RF}$ $K$-point IFFT's. Note that our model assumes that all subcarriers are used and, therefore, the data block length is equal to the number of subcarriers. A cyclic prefix of length $D$ is then added to the symbol blocks before applying the $N_\mathrm{BS} \times N_\mathrm{RF}$ RF precoding $\bF_\mathrm{RF}$. It is important to emphasize here that the RF precoding matrix $\bF_\mathrm{RF}$ is the same for \textit{all} subcarriers. This means that the RF precoder is assumed to be frequency flat while the baseband precoders can be different for each subcarrier. The discrete-time transmitted complex baseband signal at subcarrier $k$ can therefore be written as
\begin{equation}
\by[k]=\bF_\mathrm{RF} \bF[k] \bs[k],
\end{equation}
where $\bs[k]$ is the $N_\mathrm{S}\times 1$ transmitted vector at subcarrier $k$, such that  $\bbE\left[\bs[k]\bs[k]^*\right] = \frac{P}{K N_\mathrm{S}} \bI_{N_\mathrm{S}}$, and $P$ is the average total transmit power. Since $\bF_\mathrm{RF}$ is implemented using analog phase shifters, its entries are of constant modulus. To reflect that, we normalize the entries  $\left|\left[\bF_\mathrm{RF}\right]_{m,n}\right|^2=1$. Further, we assume that the angles of the analog phase shifters are quantized and have a finite set of possible values. With these assumptions, $\left[\bF_\mathrm{RF}\right]_{m,n}= e^{\j \phi_{m,n}}$, where $\phi_{m.n}$ is a quantized angle. The hybrid precoders are assumed to have a unitary power constraint,i.e., they meet $\bF_\mathrm{RF} \bF[k] \in \cU_{N_\mathrm{BS} \times N_\mathrm{S}}$, with the set of semi-unitary matrices $\cU_{N_\mathrm{BS} \times N_\mathrm{S}}=\left\{\bU \in \mathbb{C}^{N_\mathrm{BS} \times N_\mathrm{S}} | \bU^* \bU= \bI \right\}$.

\begin{figure}[t]
	\centerline{
		\includegraphics[scale=.25]{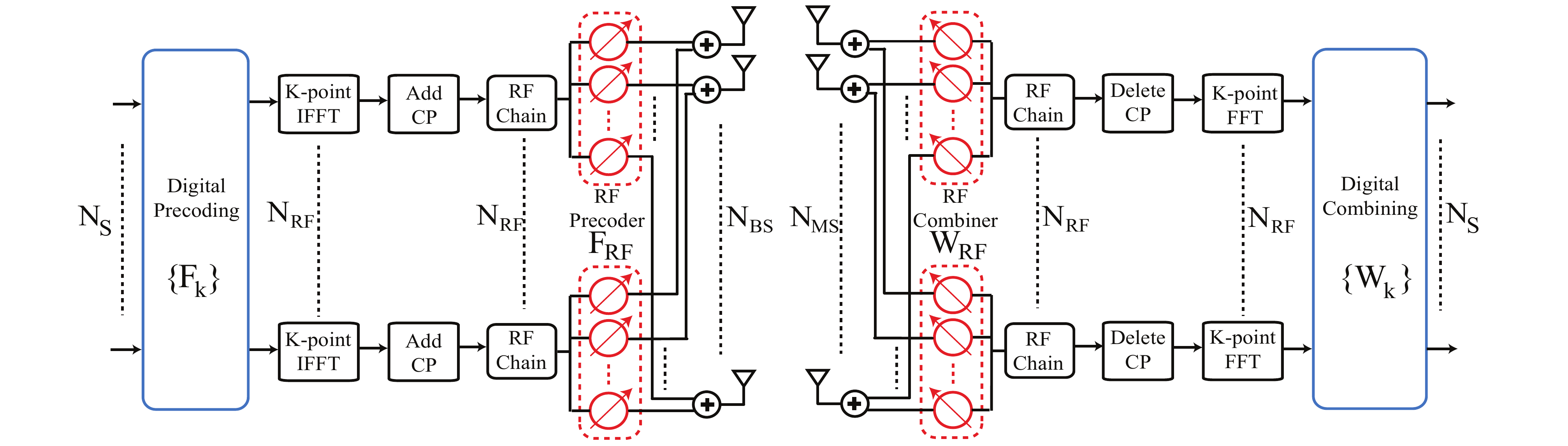}
	}
	\caption{A block diagram of the OFDM based BS-MS transceiver that employs hybrid analog/digital precoding.}
	\label{fig:Model}
\end{figure}

At the MS, assuming perfect carrier and frequency offset synchronization, the received signal is first combined in the RF domain using the $N_\mathrm{MS} \times N_\mathrm{RF}$ combining matrix $\bW_\mathrm{RF}$. Then, the cyclic prefix is removed, and the symbols are returned back to the frequency domain where the symbols at each subcarrier $k$ are combined using the $N_\mathrm{RF} \times N_\mathrm{S}$ digital combining matrix $\bW[k]$. Denoting the  $N_\mathrm{MS} \times N_\mathrm{BS}$ channel matrix at subcarrier $k$ as $\bH[k]$, the received signal at subcarrier $k$ after processing can be then expressed as
\begin{equation}
\by[k]=\bW[k]^* \bW_\mathrm{RF}^* \bH[k] \bF_\mathrm{RF} \bF[k] s[k]+ \bW[k]^* \bW_\mathrm{RF}^* \bn[k],
\label{eq:processed}
\end{equation}
where  $\bn[k] \sim \cN(\boldsymbol{0}, \sigma_\mathrm{N}^2 \bI)$ is a Gaussian noise vector.

\section{Problem Formulation} \label{sec:Form}
The paper objective is to develop a low-complexity hybrid precoding design to maximize the achievable system spectral efficiency. Given the system model in \sref{sec:Model}. For simplicity of exposition, we will assume that the receiver can perform optimal nearest neighbor decoding based on the $N_\mathrm{MS}$-dimensional received signal with fully digital hardware. This allows decoupling the transceiver design problem, and focusing on the hybrid precoders design to maximize the mutual information of the system \cite{ElAyach2014}, defined as 
\begin{align}
& \cI\left(\bF_\mathrm{RF}, \left\{\bF[k]\right\}_{k=1}^K\right)= \nonumber\\
& \frac{1}{K} \sum_{k=1}^{K} \log_2 \left|\bI_{N_\mathrm{MS}}+\frac{\rho}{N_\mathrm{S}} \bH[k] \bF_\mathrm{RF} \bF[k]  \bF[k]^* \bF_\mathrm{RF}^*  \bH[k]^*  \right|,
\label{eq:MI}
\end{align}
where $\rho=\frac{P}{K \sigma^2}$ is the SNR.  As combining with fully digital hardware is not a practical mmWave solution, the hybrid combining design problem needs also to be considered. The design ideas that will be given in this paper for the hybrid precoders, however, provide direct tools for constructing the hybrid combining matrices, $\bW_\mathrm{RF}$, $\left\{\bW[k]\right\}_{k=1}^K$, and is therefore omitted due to space limitations.

If the RF beamforming vectors are taken from a codebook $\cF_\mathrm{RF}$ that captures the RF hardware constraints, then the maximum mutual information under the given hybrid precoding model is
\begin{equation}
\begin{aligned}
\cI^{\star}_\mathrm{HP} =  & \underset{\bF_\mathrm{RF}, \left\{\bF[k]\right\}_{ k=1}^K}  \max
& &\cI\left(\bF_\mathrm{RF}, \left\{\bF[k]\right\}_{k=1}^K\right) \\
& \hspace{20pt} \text{s.t.}
& & \hspace{-22pt} \left[\bF_\mathrm{RF }\right]_{:,r} \in \cF_\mathrm{RF}, \ \ r=1, ..., N_\mathrm{RF} \\
&&& \hspace{-20pt} \bF_\mathrm{RF} \bF[k] \in \cU_{N_\mathrm{BS} \times N_\mathrm{RF}}, \ \ k=1, 2, ..., K.
\label{eq:Opt_Feedback}
\end{aligned}
\end{equation}

One challenge of the hybrid precoding design to solve the optimization problem in \eqref{eq:Opt_Feedback} is the coupling between baseband and RF precoders that arises in the power constraint (the second constraint of \eqref{eq:Opt_Feedback}). In the following proposition, we show that the baseband precoders can be written optimally as a function of the RF precoders. 
\begin{proposition}
	Define the SVD decompositions of the matrices $\bH[k]=\bU[k] \boldsymbol{\Sigma}[k] \bV[k]^*$ and  $\boldsymbol{\Sigma}[k] \bV[k]^* \bF_\mathrm{RF} \left(\bF_\mathrm{RF}^* \bF_\mathrm{RF} \right)^{-\frac{1}{2}}$ $=\overline{\bU}[k] \overline{\boldsymbol{\Sigma}}[k] \overline{\bV}[k]^*$, then the baseband precoders $\left\{\bF[k]\right\}_{k=1}^K$ that solve  \eqref{eq:Opt_Feedback} are given by \label{prop:Opt_UP}
	\begin{equation}
	\bF[k]^\star= \left(\bF_\mathrm{RF}^* \bF_\mathrm{RF} \right)^{-\frac{1}{2}} \left[\overline{\bV}[k]\right]_{:,1:N_\mathrm{S}}, \ \ k=1, 2, ..., K. \label{eq:Opt_BB_UP}
	\end{equation} 
\end{proposition} 
\begin{proof}
	The proof follows using change of variables. It is omitted due to space limitation, but available in the journal version \cite{Alkhateeb2015a}.
\end{proof}

Given proposition \ref{prop:Opt_UP}, the optimal hybrid precoding based mutual information can be given by making an exhaustive search over only the RF precoding codebook. To avoid this search. we propose efficient greedy hybrid precoding algorithms in the following sections. 

\section{Greedy Hybrid Precoding} \label{sec:Greedy}
 A natural greedy approach to construct the hybrid precoder is to iteratively select the $N_\mathrm{RF}$ RF beamforming vectors from the codebook $\cF_\mathrm{RF}$ to maximize the mutual information. In this paper, we call this  the direct greedy hybrid precoding (DG-HP) algorithm.  Let the $N_\mathrm{BS}\times (i-1)$ matrix $\bF_\mathrm{RF}^{(i-1)}$ denote the RF precoding matrix at the end of the $(i-1)$th iteration. Then by leveraging the optimal baseband precoder structure in \eqref{eq:Opt_BB_UP}, the objective of the $i$th iteration is to select $\bff^\mathrm{RF}_{n} \in \cF_\mathrm{RF}$ that solves
\begin{align}
& \cI_\mathrm{HP}^{(i)}=\underset{\bff^\mathrm{RF}_n \in \cF_\mathrm{RF}}{\max}\frac{1}{K} \sum_{k=1}^K \sum_{\ell=1}^{i} \log_2\left(1 + \frac{\rho}{N_\mathrm{RF}} \right. \nonumber \\
& \hspace{0pt} \times \left.  \lambda_\ell\left(\bH\left[k\right] \hat{\bF}_\mathrm{RF}^{(i,n)} \left({\hat{\bF}_\mathrm{RF}^{{(i,n)}^*}}\hat{\bF}_\mathrm{RF}^{(i,n)}\right)^{-1} {\hat{\bF}_\mathrm{RF}^{{(i,n)}^*}} \bH\left[k\right]^* \right)\right),
\label{eq:DG_HP}
\end{align}
with $\hat{\bF}_\mathrm{RF}^{(i,n)}=\left[\bF_\mathrm{RF}^{(i-1)}, \bff^\mathrm{RF}_n\right]$. The best vector $\bff^\mathrm{RF}_{n^\star}$ will be then added to the RF precoding matrix to form $\bF_\mathrm{RF}^{(i)}=\left[\bF_\mathrm{RF}^{(i-1)},  \bff^\mathrm{RF}_{n^\star}\right]$. The achievable mutual information with this algorithm is then $\cI_\mathrm{HP}^\mathrm{DG-HP}=\cI_\mathrm{HP}^{(N_\mathrm{RF})}$. The main limitation of this algorithm is that it still requires an exhaustive search over $\cF_\mathrm{RF}$ and eigenvalues calculation in each iteration. In the next section, we will make a first step towards developing a low-complexity algorithm that has a similar (or very close) performance to this DG-HP algorithm. 
\section{Gram-Schmidt Greedy Hybrid Precoding} \label{sec:GS_Main}
In hybrid analog/digital precoding architectures, the effective channel seen at the baseband is through the RF precoders lens. This gives the intuition that it is better for the RF beamforming vectors to be orthogonal (or close to orthogonal), as this physically means that the effective channel will have a better coverage over the dominant subspaces belonging to the actual channel matrix. This intuition is also confirmed by the structure of the optimal baseband precoder discussed in Proposition \ref{prop:Opt_UP}, as the overall matrix $\bF_\mathrm{RF}\left(\bF_\mathrm{RF}^* \bF_\mathrm{RF}\right)^{-\frac{1}{2}}$ has a semi-unitary structure. This note means that in each iteration $i$ of the greedy hybrid precoding algorithm in \eqref{eq:DG_HP} with a selected codeword $\bff^\mathrm{RF}_{n^\star}$, the additional mutual information gain over the previous iterations is due to the contribution of the component of $\bff^\mathrm{RF}_{n^\star}$ that is orthogonal on the existing RF precoding matrix $\bF_\mathrm{RF}^{(i-1)}$. This is similar to the greedy user scheduling in MIMO broadcast channels based on the orthogonal channel components \cite{Yoo2006}, but in a different context. Based on that, we modify the DG-HP algorithm by adding a Gram-Schmidt orthogonalization step in each iteration $i$ to project the candidate beamforming codewords on the orthogonal complement of the subspace spanned by the selected codewords in $\bF_\mathrm{RF}^{(i-1)}$. This can be simply done by multiplying the candidate vectors by the projection matrix ${\bP^{(i-1)}}^{\perp}=\left(\bI_i-\bF_\mathrm{RF}^{(i-1)}\left(\bF_\mathrm{RF}^{(i-1)^*} \bF_\mathrm{RF}^{(i-1)}\right)^{-1} \bF_\mathrm{RF}^{(i-1)^*}\right)$. Given the optimal precoder design in \eqref{eq:Opt_BB_UP}, the mutual information at the $i$th iteration of the modified Gram-Schmidt hybrid precoding (GS-HP) algorithm can be written as
\begin{small}
\begin{align}
\overline{\cI}_\mathrm{HP}^{(i)} &=\underset{\bff^\mathrm{RF}_n \in {\cF}_\mathrm{RF}}{\max}\frac{1}{K} \sum_{k=1}^K \sum_{\ell=1}^{i} \log_2\left(1 + \frac{\rho}{N_\mathrm{RF}} \lambda_\ell\left(\bH\left[k\right] \overline{\bF}_\mathrm{RF}^{(i,n)}\right. \right. \nonumber \\ & \hspace{40pt} \times \left. \left. \left({\overline{\bF}_\mathrm{RF}^{(i,n)}}^* \overline{\bF}_\mathrm{RF}^{(i,n)}\right)^{-1} {\overline{\bF}_\mathrm{RF}^{(i,n)}}^* \bH\left[k\right]^* \right)\right),\label{eq:GS_HP}\\
&\stackrel{(a)}{=} \underset{\bff^\mathrm{RF}_n \in {\cF}_\mathrm{RF}}{\max}\frac{1}{K} \sum_{k=1}^K \sum_{\ell=1}^{i} \log_2\left(1 + \frac{\rho}{N_\mathrm{RF}} \lambda_\ell\left( \bT^{(i-1)} \right. \right. \nonumber \\
& \hspace{20pt} \left. \left. + \bH[k] {\bP^{(i-1)}}^{\perp} \bff^\mathrm{RF}_n  {\bff^\mathrm{RF}_n}^* {{\bP^{(i-1)}}^{\perp}}^* \bH^*[k] \right)\right),
\label{eq:GS_HP2}
\end{align}
\end{small}
with $\bT^{(i-1)} \hspace{-2pt} = \hspace{-2pt} \bH\left[k\right] \bF_\mathrm{RF}^{(i-1)}\hspace{-3pt}\left({\bF_\mathrm{RF}^{(i-1)}}^* \bF_\mathrm{RF}^{(i-1)}\right)^{\hspace{-3pt}-1} {\bF_\mathrm{RF}^{(i-1)}}^* \bH\left[k\right]^*$, and $\overline{\bF}_\mathrm{RF}^{(i,n)}=\left[\bF_\mathrm{RF}^{(i-1)}, {\bP^{(i-1)}}^{\perp} \bff^\mathrm{RF}_n\right]$. Note that $\bT^{(i-1)}$ is a constant matrix at iteration $i$, and (a) follows from the Gram-Schmidt orthogonalization which allows the matrix $\overline{\bF}_\mathrm{RF}^{(i,n)}\left({\overline{\bF}_\mathrm{RF}^{(i,n)}}^* \overline{\bF}_\mathrm{RF}^{(i,n)}\right)^{-\frac{1}{2}}$ to be written as $\left[\bF_\mathrm{RF}^{(i-1)}\left({\bF_\mathrm{RF}^{(i-1)}}^* \bF_\mathrm{RF}^{(i-1)}\right)^{-\frac{1}{2}},  {\bP^{(i-1)}}^{\perp} \bff_n^\mathrm{RF} \right]$. Hence, the eigenvalues calculation in \eqref{eq:GS_HP2} can be calculated as a rank-1 update of the previous iteration eigenvalues, which reduces the overall complexity \cite{Br2006}. The best vector $\bff^\mathrm{RF}_{n^\star}$ will be then added to the RF precoding matrix to form $\bF_\mathrm{RF}^{(i)}=\left[\bF_\mathrm{RF}^{(i-1)},  \bff^\mathrm{RF}_{n^\star}\right]$. At the end of the $N_\mathrm{RF}$ iterations, we get $\cI_\mathrm{HP}^\mathrm{GS-HP}=\overline{\cI}_\mathrm{HP}^{(N_\mathrm{RF})}$.  In the following proposition, we prove that this Gram-Schmidt hybrid precoding algorithm is exactly equivalent to the DG-HP algorithm.

\begin{proposition}
	The achieved mutual information of the direct greedy hybrid precoding algorithm in \eqref{eq:DG_HP} and the Gram-Schmidt hybrid precoding algorithm in \eqref{eq:GS_HP} are exactly equal, i.e., $\cI_\mathrm{HP}^\mathrm{DG-HP}=\cI_\mathrm{HP}^\mathrm{GS-HP}$.\label{prop:GS}
\end{proposition}
\begin{proof}
	See  \sref{app:GS}.
\end{proof}
\section{Approximate Gram-Schmidt Based Greedy Hybrid Precoding} \label{subsec:Approx_GS}

\begin{algorithm} [!t]                     
	\caption{Approximate Gram-Schmidt Greedy Hybrid Precoding}          
	\label{alg:GS_HP}                           
	\begin{algorithmic} 
		\State \textbf{Initialization}
		\State 1) \begin{varwidth}[t]{\linewidth}  Construct $\boldsymbol{\Pi} = \tilde{\boldsymbol{\Sigma}}_\mathrm{\bH} \tilde{\bV}_\mathrm{\bH}$, with $\tilde{\boldsymbol{\Sigma}}_\mathrm{\bH}=\left[\tilde{\boldsymbol{\Sigma}}_\mathrm{1}, ..., \tilde{\boldsymbol{\Sigma}}_\mathrm{K}\right]$ and $\tilde{\bV}_\mathrm{\bH}=\left[\tilde{\bV}_\mathrm{1}, ..., \tilde{\bV}_\mathrm{K}\right]$. Set $\bF_\mathrm{RF}=$ Empty Matrix. Set $\bA_\mathrm{CB}=\left[{\bff}_1^\mathrm{RF}, ...,{\bff}^\mathrm{RF}_{N_\mathrm{CB}^\mathrm{v}} \right]$, where ${\bff}^\mathrm{RF}_n, n=1, ..., {N_\mathrm{CB}^\mathrm{v}}$ are the codewords in ${\cF}_\mathrm{RF}$.
		\end{varwidth}
		\State \textbf{RF Precoder Design}
		\State 2) \text{For}{ $i, i = 1, ..., N_\mathrm{RF}$}
		\State  $\hspace{20pt}$ a) $ \boldsymbol\Psi = \boldsymbol\Pi^*{\bA}_\mathrm{CB}$
		\State  $\hspace{20pt}$ b) $n^\star=\arg\max_{n=1,2,..N_\mathrm{CB}^\mathrm{v}} \left\|\left[\boldsymbol\Psi\right]_{:,n}\right\|_2$.
		\State  $\hspace{20pt}$ c) $\bF_\mathrm{RF}^{(i)} = \left[\bF_\mathrm{RF}^{(i-1)} \bff_{n^\star}^\mathrm{RF}\right]$
		\State $\hspace{20pt}$ d) $\boldsymbol\Pi=\boldsymbol\Pi \left(\bI_i-\bF_\mathrm{RF}^{(i)}\left(\bF_\mathrm{RF}^{(i)^*} \bF_\mathrm{RF}^{(i)}\right)^{-1}\bF_\mathrm{RF}^{(i)^*}\right)$
		\State \textbf{Digital Precoder Design}
		\State 3) \begin{varwidth}[t]{\linewidth} $\bF[k]= \bF_\mathrm{RF}^{(N_\mathrm{RF})} \left(\bF_\mathrm{RF}^{(N_\mathrm{RF})^*} \bF_\mathrm{RF}^{(N_\mathrm{RF})} \right)^{-\frac{1}{2}} \left[\overline{\bV}[k]\right]_{:,1:N_\mathrm{S}}, k=1, ..., K$,  with $\overline{\bV}[k]$ defined in \eqref{eq:Opt_BB_UP}. \end{varwidth}
	\end{algorithmic}
\end{algorithm}

The main advantage of the Gram-Schmidt hybrid precoding design in \sref{sec:GS_Main} is that it leads to a near-optimal low-complexity design of the frequency selective hybrid precoding as will be discussed in this section. Given the optimal baseband precoding solution in \eqref{eq:Opt_BB_UP}, the mutual information at the $i$th iteration in \eqref{eq:GS_HP} can be written as
\begin{small}
\begin{align}
\overline{\cI}_\mathrm{HP}^{(i)}& = \underset{\bff^\mathrm{RF}_n \in {\cF}_\mathrm{RF}}{\max}\frac{1}{K} \sum_{k=1}^K \sum_{\ell=1}^{i} \log_2\left(1 + \frac{\rho}{N_\mathrm{RF}} \lambda_\ell\left(\bH\left[k\right] \overline{\bF}_\mathrm{RF}^{(i,n)} \right. \right. \nonumber \\
& \hspace{40pt} \times \left. \left. \left({\overline{\bF}_\mathrm{RF}^{(i,n)}}^* \overline{\bF}_\mathrm{RF}^{(i,n)}\right)^{-1} {\overline{\bF}_\mathrm{RF}^{(i,n)}}^* \bH\left[k\right]^* \right)\right),\\
&\hspace{-20pt}\stackrel{(a)}{\geq} \underset{\bff^\mathrm{RF}_n \in {\cF}_\mathrm{RF}}{\max} \frac{1}{K} \sum_{k=1}^K \sum_{\ell=1}^{i} \log_2  \left( 1+\frac{\rho}{N_\mathrm{S}}  \lambda_\ell\left(\tilde{\boldsymbol{\Sigma}}[k] \tilde{\bV}^*[k] \right. \right. \nonumber \\
&\hspace{-20pt}  \times \left. \left.  \overline{\bF}_\mathrm{RF}^{(i,n)}\left({\overline{\bF}_\mathrm{RF}^{(i,n)}}^* \overline{\bF}_\mathrm{RF}^{(i,n)}\right)^{-1} {\overline{\bF}_\mathrm{RF}^{(i,n)}}^* \tilde{\bV}[k] \tilde{\boldsymbol{\Sigma}}^*[k] \right)\right), \\
&\hspace{-20pt} \stackrel{(b)}{\approx} \frac{1}{K} \sum_{k=1}^K \left(\log_2\left|\bI +\frac{\rho}{N_\mathrm{S}}  \tilde{\boldsymbol{\Sigma}}[k]^2 \right| - \text{tr} \left(\tilde{\boldsymbol{\Sigma}}[k]\right) \right)  \nonumber\\  
&\hspace{-20pt} + \underset{\bff^\mathrm{RF}_n \in {\cF}_\mathrm{RF}}{\max} \frac{1}{K} \sum_{k=1}^K \left\|\tilde{\boldsymbol{\Sigma}}[k] \tilde{\bV}[k]^* \overline{\bF}_\mathrm{RF}^{(i,n)} \hspace{-3pt} \left({\overline{\bF}_\mathrm{RF}^{(i,n)}}^* \overline{\bF}_\mathrm{RF}^{(i,n)}\right)^{\hspace{-3pt}- \frac{1}{2}} \right\|_F^2
\end{align}
\end{small}%
where (a) is by considering only the first $N_\mathrm{S}$ dominant singular values of $\bH[k]$, $\tilde{\boldsymbol{\Sigma}}[k]=\left[{\boldsymbol{\Sigma}}[k]\right]_{:,1:N_\mathrm{S}}$, $\tilde{\bV}[k]=\left[{\bV}[k]\right]_{:,1:N_\mathrm{S}}$, and (b) follows from using the large mmWave MIMO approximations used in \cite{ElAyach2014}. The objective of the $i$th iteration is then to solve
\begin{small}\begin{align}
\bff^\mathrm{RF}_{n^\star}&=\underset{{\bff}^\mathrm{RF}_n \in {\cF}_\mathrm{RF}}{\arg\max}\frac{1}{K} \sum_{k=1}^K \left\|\tilde{\boldsymbol{\Sigma}}[k] \tilde{\bV}[k]^* \overline{\bF}_\mathrm{RF}^{(i,n)} \hspace{-3pt} \left({\overline{\bF}_\mathrm{RF}^{(i,n)}}^* \overline{\bF}_\mathrm{RF}^{(i,n)}\right)^{\hspace{-3pt}-\frac{1}{2}} \right\|_F^2 \\
& =\underset{{\bff}^\mathrm{RF}_n \in {\cF}_\mathrm{RF}}{\arg\max}\left\|\tilde{\boldsymbol{\Sigma}}_\mathrm{\bH} \tilde{\bV}_\mathrm{\bH}^* \bF^{(i-1)}_\mathrm{RF} \left(\bF^{(i-1)^*}_\mathrm{RF} \bF^{(i-1)}_\mathrm{RF}\right)^{-\frac{1}{2}} \right\|_F^2, \\
&\stackrel{(a)}{=} \underset{{\bff}^\mathrm{RF}_n \in {\cF}_\mathrm{RF}}{\arg\max}\left\|\tilde{\boldsymbol{\Sigma}}_\mathrm{\bH} \tilde{\bV}_\mathrm{\bH}^* {\bP^{(i-1)}}^\perp {\bff}_n^\mathrm{RF}\right\|_2^2,  \label{eq:Opt_GS_Simple}
\end{align}\end{small}
\hspace{-5pt} where $\tilde{\boldsymbol{\Sigma}}_\mathrm{\bH}=\left[\tilde{\boldsymbol{\Sigma}}[1], ..., \tilde{\boldsymbol{\Sigma}}[K]\right]$, $\tilde{\bV}_\mathrm{\bH}=\left[\tilde{\bV}[1], ..., \tilde{\bV}[K]\right]$, and (a) is a result of the Gram-Schmidt processing as described in \sref{sec:GS_Main}. The problem in \eqref{eq:Opt_GS_Simple} is simple to solve with just a maximum projection step. We call this algorithm the approximate Gram-Schmidt hybrid precoding (Approximate GS-HP) algorithm. As shown in  Algorithm \ref{alg:GS_HP}, the developed algorithm sequentially build the RF and baseband precoding matrices in two separate stages. First, the RF beamforming vectors are iteratively selected to solve \eqref{eq:Opt_GS_Simple}. Then, the baseband precoder is optimally designed according to \eqref{eq:Opt_BB_UP}. Despite its sequential design of the RF and baseband precoders, which reduces the complexity when compared with prior solutions that mostly depend on the joint design of the baseband and RF precoding matrices \cite{ElAyach2014,Alkhateeb2014}, Algorithm \ref{alg:GS_HP} achieves a significant gain over prior solutions, and gives a very close performance to the optimal solution given by exhaustive search, as will be shown in \sref{sec:Results}.   

\section{Simulation Results} \label{sec:Results}
In this section, we evaluate the performance of the proposed algorithm using numerical simulations. We adopt a wideband mmWave channel model that consists of $L=6$ clusters. The center AoAs/AoDs of the $L$ clusters $\theta_\ell, \phi_\ell$ are assumed to be uniformly distributed in $[0, 2 \pi)$. Each cluster has  $R_\ell=5$ rays  with Laplacian distributed AoAs/AoDs \cite{Forenza2007,ElAyach2014}, and angle spread of $10^\mathrm{o}$. The number of system subcarriers $K$ equals $512$, and the cyclic prefix length is $D=128$, which is similar to 802.11ad \cite{11ad}. The paths delay is uniformly distributed in $[0, D T_\mathrm{s}]$. Both the BS and MS have ULAs with $N_\mathrm{RF}=3$. 
\begin{figure}[t]
	\centerline{
		\includegraphics[width=1\columnwidth]{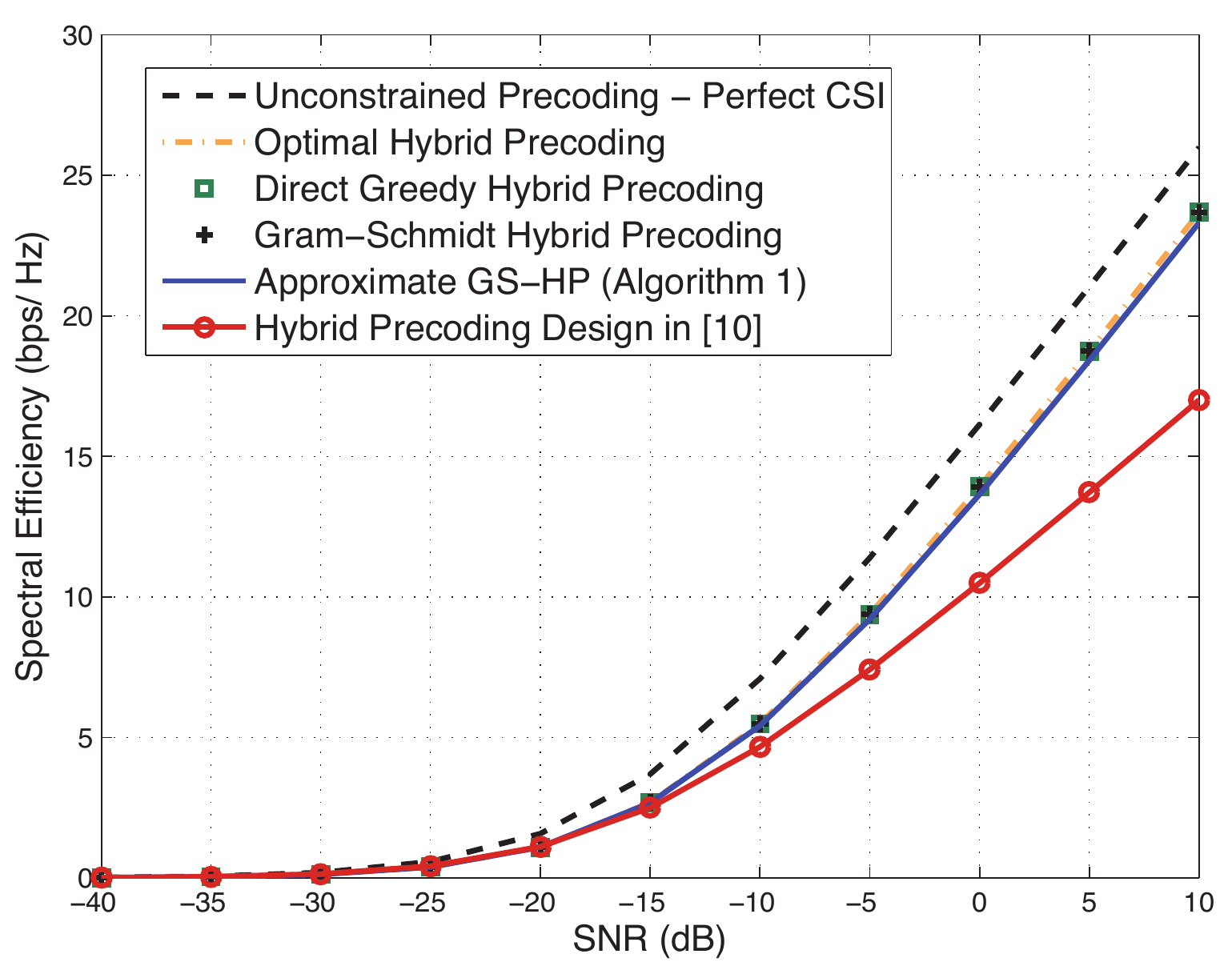}
	}
	\caption{The performance of the approximate Gram-Schmidt hybrid precoding design in Algorithm \ref{alg:GS_HP} compared with the optimal hybrid precoding solution, the unconstrained SVD solution, and the prior work in \cite{ElAyach2014}. The system has  $N_\mathrm{BS}=32$ antennas, $N_\mathrm{MS}=16$ antennas, and $N_\mathrm{S}=N_\mathrm{RF}=3$.}
	\label{fig:Fig4}
\end{figure}
\vspace{-10pt}

In \figref{fig:Fig4}, we validate the result in Proposition \ref{prop:GS}, in addition to evaluating the approximate Gram-Schmidt based hybrid precoding algorithm. The spectral efficiencies achieved by these greedy algorithms are compared with the optimal hybrid precoding design given by the exhaustive search over the RF codebooks. The rates are also compared with the prior solution in \cite{ElAyach2014}. For a fair comparison, we assume that each RF beamforming vector is selected from a beamsteering codebook with a size $N_\mathrm{CB}=64$. First, \figref{fig:Fig4} shows that the direct greedy and Gram-Schmidt based hybrid precoding algorithms achieve exactly the same performance which verifies Proposition \ref{prop:GS}. Despite its low-complexity, the developed approximate Gram-Schmidt hybrid precoding design in Algorithm \ref{alg:GS_HP} achieves very close performance to the exhaustive-search based optimal solution. We emphasize here that any hybrid precoding design can not perform better that the shown optimal hybrid precoding solution with the considered RF codebook, which confirms the near-optimal result of the proposed algorithm. This is also clear from the considerable gain obtained by the proposed algorithm compared with the prior solution in \cite{ElAyach2014}. Also, it is worth mentioning that the developed hybrid precoding algorithms in this paper can be applied to any large MIMO system (not specifically mmWave systems). 
\section{Conclusion}\label{sec:Conc}
In this paper, we investigated hybrid precoding design for wideband mmWave systems. First, we derived the optimal hybrid precoding design that maximizes the achievable mutual information for any given RF codebook, and showed that the optimal baseband structure can be decomposed into an RF precoder dependent matrix and a unitary matrix. Second, we developed a novel greedy hybrid precoding algorithm based on Gram-Schmidt orthogonalization. Thanks to this Gram-Schmidt orthogonalization, we showed that only sequential design of the RF and baseband precoders is required to achieve the same performance of more sophisticated algorithms that requires a joint design of the RF and baseband precoders in each step. Simulation results illustrated that the proposed precoding algorithms improve over prior work and stay within a small gap from the unconstrained perfect channel knowledge solutions. 
\section{Proof of Proposition 2}\label{app:GS}
\begin{proof}
To prove that $\cI_\mathrm{HP}^\mathrm{GS-HP}=\cI_\mathrm{HP}^\mathrm{DG-HP}$, it is sufficient to prove that $\bF_\mathrm{RF}^{(N_\mathrm{RF})}$ of the GS-HP and DG-HP algorithms are equal. To do that, we will show that both the algorithms choose the same RF beamforming vector in each iteration, i.e., $\bF_\mathrm{RF}^{(i)}$ is equal for $i=1,...,N_\mathrm{RF}$. This can be proved using  mathematical induction as follows. At the first iteration, the two algorithms do the exhaustive search over the same codebook $\cF_\mathrm{RF}$, and consequently select the same beamforming vectors. Now, suppose that the two algorithms reach the same RF precoding matrix $\bF_\mathrm{RF}^{(i-1)}$ at iteration $i-1$, we need to prove that they both select the same RF beamforming vector at iteration $i$, i.e., we need to prove that both \eqref{eq:DG_HP} and \eqref{eq:GS_HP} choose beamforming vectors with the same index. To prove that, it is enough to show that the contributions of the $n$th beamforming vector $\bff_n^\mathrm{RF}$ from $\cF_\mathrm{RF}^{\mathrm{v}}$ in \eqref{eq:DG_HP} and \eqref{eq:GS_HP} are equal.Given the optimal baseband precoder in \eqref{eq:Opt_BB_UP}, and denoting the SVD of $\hat{\bF}_\mathrm{RF}^{(i,n)}$ as  $\hat{\bF}_\mathrm{RF}^{(i,n)}=\hat{\bU}_\mathrm{RF}^{(i,n)} \hat{\boldsymbol{\Sigma}}_\mathrm{RF}^{(i,n)}\hat{\bV}^{(i,n)^*}_\mathrm{RF}$, equation \eqref{eq:DG_HP} can be written as
\begin{small}
\begin{align}
& \sum_{\ell=1}^{i} \log_2  \left( 1+\frac{\rho}{N_\mathrm{S}}  \lambda_\ell\left(\bH[k] \hat{\bF}^{(i,n)}_\mathrm{RF} \left(\hat{\bF}^{(i,n)^*}_\mathrm{RF} \hat{\bF}^{(i,n)}_\mathrm{RF}\right)^{-1}  \right. \right. \nonumber \\
& \hspace{120pt} \times \left. \left. \vphantom{\left(\hat{\bF}^{(i,n)^*}_\mathrm{RF} \hat{\bF}^{(i,n)}_\mathrm{RF}\right)^{-1} } \hat{\bF}^{(i,n)^*}_\mathrm{RF}  \bH[k]^* \right)\right), \label{eq:GS_equal}\\
& = \sum_{\ell=1}^{i} \log_2  \left( 1+\frac{\rho}{N_\mathrm{S}}  \lambda_\ell\left(\bH[k] \hat{\bU}_\mathrm{RF}^{(i,n)} \hat{\bU}_\mathrm{RF}^{(i,n)^*} \bH[k]^* \right)\right).
\end{align}
\end{small}

\noindent Equation \eqref{eq:GS_HP} can be similarly written, but with $\hat{\bU}_\mathrm{RF}^{(i,n)}$ replaced by $\overline{\bU}_\mathrm{RF}^{(i,n)}$ where $\overline{\bF}_\mathrm{RF}^{(i,n)}=\overline{\bU}_\mathrm{RF}^{(i,n)} \overline{\boldsymbol{\Sigma}}_\mathrm{RF}^{(i,n)} \overline{\bV}^{(i,n)^*}_\mathrm{RF}$. Hence, we  need to prove that $\hat{\bU}_\mathrm{RF}^{(i,n)} \hat{\bU}_\mathrm{RF}^{(i,n)^*}=\overline{\bU}_\mathrm{RF}^{(i,n)} \overline{\bU}_\mathrm{RF}^{(i,n)^*}$. Let $\overline{\bff}_n=\bP^{(i-1)^\perp} \bff_n$ denote the last column of $\overline{\bF}_\mathrm{RF}^{(i,n)}$. As $\overline{\bff}_n$ is a result of successive Gram-Schmidt operations, we can write $\bff_n=\overline{\bff}_n+\bF_\mathrm{RF}^{(i-1)} \boldsymbol{\alpha}_n$, where $\boldsymbol{\alpha}_n$ is a vector results from the Gram-Schmidt process. Consequently,  $\hat{\bF}_\mathrm{RF}^{(i,n)}$ can be written as $\hat{\bF}_\mathrm{RF}^{(i,n)}=\overline{\bF}_\mathrm{RF}^{(i,n)}\bE_C$, where $\bE_C=\left[\begin{array}{cc} \bI & \boldsymbol{\alpha} \\ \boldsymbol{0}^T & 1\end{array}\right]$ is an elementary column operation matrix. Now, we note that $\hat{\bU}_\mathrm{RF}^{(i,n)} \hat{\bU}_\mathrm{RF}^{(i,n)^*}=\hat{\bF}_\mathrm{RF}^{(i,n)} \hat{\bF}_\mathrm{RF}^{(i,n)^{\dagger}}=\overline{\bF}_\mathrm{RF}^{(i,n)} \bE_C \bE_C^{-1} \overline{\bF}_\mathrm{RF}^{(i,n)^{\dagger}}= \overline{\bU}_\mathrm{RF}^{(i,n)} \overline{\bU}_\mathrm{RF}^{(i,n)^*}$, as $\bE_C$ is an $i \times i$ full-rank matrix.
\end{proof}

\bibliographystyle{IEEEtran}

\end{document}